\documentclass[runningheads]{llncs}

\usepackage{listings}
\usepackage{algorithm}
\usepackage{algorithmic}
\usepackage{balance}
\usepackage{mathpartir}
\usepackage{proof}
\usepackage{times}
\usepackage[english]{babel}
\usepackage{xcolor}
\usepackage{pgf}
\usepackage[T1]{fontenc}
\usepackage[scaled]{beramono}
\usepackage{booktabs}
\usepackage[hyphens]{url}
\usepackage{hyperref}
\usepackage[hyphenbreaks]{breakurl}
\usepackage{bibentry}
\usepackage[normalem]{ulem}
\usepackage{tikz}
\usetikzlibrary{calc,intersections,through,backgrounds,matrix,patterns}
\usepackage{pgfplots}
\usepackage{etoolbox}
\usepackage[utf8]{inputenc}

\usepackage{xifthen}

\usepackage{graphicx}
\usepackage[]{xparse}
\usepackage[]{xspace}
\usepackage{amsmath}
\usepackage{amssymb}
\usepackage{mathtools}
\usepackage{stmaryrd}
\usepackage[T1]{fontenc}
\usepackage{cleveref}

\tikzset{every tree node/.style={minimum width=2em},
         blank/.style={draw=none},
         edge from parent/.style=
         {draw, color=gray, edge from parent path={(\tikzparentnode) -- (\tikzchildnode)}},
         level distance=1.5cm}

\newcommand{\ie}{i.e., }
\newcommand{\eg}{e.g., }

\newcommand{\true}{1\xspace}
\newcommand{\false}{0\xspace}

\newtoggle{REPORT}
\newtoggle{DEBUG}
\toggletrue{REPORT}
\toggletrue{DEBUG}

\definecolor{reportcolor}{HTML}{188781}
\definecolor{papercolor}{HTML}{b37400}

\newcommand{\reportmode}{\toggletrue{REPORT}\togglefalse{DEBUG}}

\newcommand{\reportonly}[1]{\iftoggle{DEBUG}{{\color{reportcolor} #1}}{\iftoggle{REPORT}{#1}{}}}
\newcommand{\paperonly}[1]{\iftoggle{DEBUG}{{\color{papercolor} #1}}{\iftoggle{REPORT}{}{#1}}}

\newcommand{\debugonly}[1]{\iftoggle{DEBUG}{#1}}

\newcommand{\reportorpaper}[2]{\iftoggle{DEBUG}{{\color{reportcolor} #1} {\color{papercolor} #2}}{\iftoggle{REPORT}{#1}{#2}}}

\newcommand{\shapename}[1]{\emph{#1}}
\newcommand{\labelname}[1]{\emph{#1}}
\newcommand{\propname}[1]{\emph{#1}}
\newcommand{\stringvalue}[1]{\textrm{"#1"}}

\definecolor{darkgray}{gray}{0.20}
\definecolor{shaclextendcolor}{HTML}{004bb3}
\definecolor{shacltranslatecolor}{gray}{0.0}
\definecolor{rdfstarextendcolor}{HTML}{00b368}
\definecolor{lightgray}{gray}{0.40}
\definecolor{meta1color}{HTML}{0066cc}
\definecolor{meta2color}{HTML}{cc0066}
\definecolor{typecolor}{gray}{0.55}
\newcommand{\hsp}{\vphantom{Ag}}

\newlength\vgap
\newlength\hgap
\newlength\kvgap

\newcommand\nodelabel[2]{{\color{black}\textbf{#1}} \labelname{#2}}
\newcommand\edgelabel[2]{{\color{black}\textbf{#1}} \labelname{#2}}
\newcommand\wikilabel[2]{{\color{lightgray}\emph{#1}}\\\labelname{#2}}
\newcommand\propertylabelstring[2]{\propname{#1}: \{\stringvalue{#2}\}}
\newcommand\propertylabelstrings[3]{\propname{#1}: \{\stringvalue{#2},\stringvalue{#3}\}}
\newcommand\propertylabelint[2]{\propname{#1}: \{#2\}}

\newcommand\propertylabeldate[2]{\propname{#1}: \{#2\}}

\newcommand\shaclextend[1]{{\color{shaclextendcolor}#1}}

\newcommand\shacltranslate[1]{{\color{shacltranslatecolor}#1}}

\usepackage{tikz}
\usetikzlibrary{arrows, positioning, fit, shapes, backgrounds, patterns,
shapes.misc,decorations.markings}

\tikzset{%
  compact/.style={
	  inner sep=0.4ex,
	  font=\sffamily\scriptsize\hsp
	},
  iri/.style={
    draw=darkgray,
    text=darkgray,
    fill=white,
    rectangle,
    align=center,
    rounded corners,
    thick,
    text centered,
    minimum width=0.5cm,
    font=\sffamily\small\hsp
	},
  irim1/.style={iri,draw=typecolor},
  irim2/.style={iri},
  irit/.style={iri,text=typecolor,draw=typecolor},
  irii/.style={iri,dashed,text=typecolor,draw=typecolor},
  literal/.style={
    draw=lightgray,
    text=lightgray,
    fill=white,
    align=left,
    rectangle,
    thick,
    text centered,
    font=\sffamily\small\hsp
	},
  literalm1/.style={literal},
  literalm2/.style={literal},
  label/.style={
     text centered,
     anchor=center,
     fill=white,
     opacity=0.98,
     text opacity=1,
     inner sep=0.2ex,
     font=\sffamily\footnotesize\hsp
	},
  arrin/.style={
    draw=darkgray,
    text=darkgray,
    arrows={stealth-},
    thick,
    font=\sffamily\footnotesize\hsp,
    pos=0.55
	},
  arrinm1/.style={arrin,dashed,text=typecolor,draw=rdfstarextendcolor},
  arrinm2/.style={arrin,dashdotted,text=meta2color,draw=meta2color},
  arrint/.style={arrin,text=typecolor,draw=typecolor},
  arrini/.style={arrin,dashed,text=typecolor,draw=typecolor},
  arrout/.style={
    draw=darkgray,
    arrows={-stealth},
    thick,
    font=\sffamily\footnotesize\hsp,
	  pos=0.45
	},
  arroutm1/.style={arrout,dashed,draw=meta1color},
  arroutm2/.style={arrout,dashdotted,draw=meta2color},
  arrouti/.style={arrout,dashed,text=typecolor,draw=typecolor},
}

\newcommand\redsout{\bgroup\markoverwith{\textcolor{meta2color}{\rule[0.5ex]{2pt}{1.0pt}}}\ULon}

\newcommand{\tikznode}[2]{\begin{tikzpicture}[baseline=-3pt]\node[#1,inner sep=0.4ex]{\scriptsize #2};\end{tikzpicture}}
\NewDocumentCommand{\gnode}{ O{iri} m}{\tikznode{#1}{#2}}

\newlength\wordWidth
\NewDocumentCommand{\tikzedgenonodes}{O{white} O{arrin} m }{%
  \setlength{\wordWidth}{\widthof{{\scriptsize #3}} + 0.6cm}%
  \begin{tikzpicture}[baseline=-3pt]%
    \node[inner sep=0](A){};%
    \node[inner sep=0,right=\wordWidth of A]{} edge[#2] node[label,fill=#1,xshift=0pt] {\scriptsize #3} (A);%
  \end{tikzpicture}%
}
\NewDocumentCommand{\tikzedgevarwidth}{ O{arrin} O{iri} O{iri} m m m }{%
  \setlength{\wordWidth}{\widthof{{\scriptsize #5}}+0.5cm}%
  \begin{tikzpicture}[baseline=-3pt]\node[#2,compact](A){#4};%
    \node[#3,compact,right=\wordWidth of A]{#6} edge[#1] node[label,xshift=-1pt] {#5} (A);%
  \end{tikzpicture}%
}
\NewDocumentCommand{\gedgenonodes}{ O{white} O{arrin} m }{\tikzedgenonodes[#1][#2]{#3}}
\NewDocumentCommand{\gedge}{ O{iri} O{iri} m m m}{\tikzedgevarwidth[arrin][#1][#2]{#3}{#4}{#5}}

\usepgfplotslibrary{fillbetween}
\pgfdeclarelayer{background}
\pgfsetlayers{background,main}
\pgfplotsset{compat=1.16}

\newcommand{\therepo}{\url{https://github.com/softlang/progs}}

\newcommand{\pshacl}{ProGS\xspace}

\newcommand{\validation}{\textrm{VALID}\xspace}

\newcommand{\valid}[3]{\normalfont\textrm{VALID}(#1,#2,#3)\xspace}

\newcommand{\shapeassignment}{\ensuremath\Sigma\xspace}

\NewDocumentCommand{\iversonleft}{}%
{\ensuremath [\,}

\NewDocumentCommand{\iversonright}{}%
{\ensuremath \,]}

\NewDocumentCommand{\iverson}{m}%
{\ensuremath \iversonleft #1 \iversonright}

\NewDocumentCommand{\iversonerr}{O{(\phi,x)} m}%
{\ensuremath [\,#2\,]_{#1}}

\newcommand{\names}[1]{%
  \ifthenelse{\isempty{#1}}%
  {\ensuremath\mathrm{Names}\xspace}%
  {\ensuremath\mathrm{Names}(#1)}%
}%

\NewDocumentCommand{\evalpath}{O{\shapeassignment} O{n} O{G} m}%
{\ensuremath\llbracket #4 \rrbracket^{#1, #2, #3}}%

\NewDocumentCommand{\evalany}{O{\shapeassignment} O{x} O{G} m}%
{\ensuremath\llbracket #4 \rrbracket^{#1, #2, #3}}%

\NewDocumentCommand{\evaln}{O{\shapeassignment} O{n} O{G} m}%
{\ensuremath\llbracket #4 \rrbracket^{#1, #2, #3}}%

\NewDocumentCommand{\evalsn}{O{G} O{S} m}%
{\ensuremath\llbracket #3 \rrbracket^{#1, #2}}%

\NewDocumentCommand{\evalse}{O{G} O{S} m}%
{\ensuremath\llbracket #3 \rrbracket^{#1, #2}}%

\NewDocumentCommand{\evale}{O{\shapeassignment} O{e} O{G} m}%
{\ensuremath\llbracket #4 \rrbracket^{#1, #2, #3}}%

\NewDocumentCommand{\evalp}{O{\shapeassignment} O{(x,k)} O{G} m}%
{\ensuremath\llbracket #4 \rrbracket^{#1, #2, #3}}%

\NewDocumentCommand{\evalv}{O{\shapeassignment} O{v} O{G} m}%
{\ensuremath\llbracket #4 \rrbracket^{#1, #2, #3}}%

\NewDocumentCommand{\evalq}{O{G} m}%
{\ensuremath\llbracket #2 \rrbracket_{#1}}%

\NewDocumentCommand{\evali}{m}%
{\ensuremath\llbracket #1 \rrbracket_{\textrm{init}}}%

\NewDocumentCommand{\nodeshape}{O{s_N} O{\phi_N} O{q_N}}%
{\ensuremath _N \langle #1,#2,#3 \rangle}

\NewDocumentCommand{\nodeshapel}{O{s_N} O{\phi_N} O{q_N}}%
{\begin{align*}
_N \langle &#1,\\
           &#2,\\
           &#3 \rangle
\end{align*}}

\newcommand{\geqleft}[2]{\geqslant_{#1}^{\leftarrow} #2}
\newcommand{\geqright}[2]{\geqslant_{#1}^{\rightarrow} #2}
\newcommand{\leqleft}[2]{\leqslant_{#1}^{\leftarrow} #2}

\newcommand{\eqleft}[2]{=_{#1}^{\leftarrow} #2}

\newcommand{\existsleft}[1]{\exists^{\leftarrow} #1}

\newcommand{\forallleft}[1]{\forall^{\leftarrow} #1}

\newcommand{\examplegraph}{G_{\textrm{office}}}

\newcommand{\targetvk}[2]{#2_{#1}}

\newcommand\utimes{\mathbin{\ooalign{$\cup$\cr%
   \hfil\raise0.42ex\hbox{$\scriptscriptstyle\times$}\hfil\cr}}}
\newcommand\bigutimes{\mathop{\ooalign{$\bigcup$\cr%
   \hfil\raise0.36ex\hbox{$\scriptscriptstyle\boldsymbol{\times}$}\hfil\cr}}}

\NewDocumentCommand{\edgeshape}{O{s_E} O{\phi_E} O{q_E}}%
{\ensuremath _E \langle #1,#2,#3 \rangle}

\NewDocumentCommand{\propertyshape}{O{s_P} O{\phi_P} O{q_P}}%
{\ensuremath _P \langle #1,#2,#3 \rangle \ }

\NewDocumentCommand{\csrule}{m m m}%
{$\rightarrow_{#1}$ & #2 & #3\\}

\lstdefinestyle{progs}{
  showstringspaces=false,
  basicstyle=\footnotesize\ttfamily,
  keywordstyle=\bfseries\color{black},
  commentstyle=\itshape\color{white!40!black},
  stringstyle=\color{white!40!black},
  morecomment=[l]{//},
  morestring=[b]",
  morekeywords={NODE,EDGE},
}

\lstdefinestyle{asp}{
  showstringspaces=false,
  basicstyle=\footnotesize\ttfamily,
  keywordstyle=\bfseries\color{black},
  commentstyle=\itshape\color{white!40!black},
  stringstyle=\color{white!40!black},
  morecomment=[l]{//},
  morestring=[b]",
  morekeywords={edge,label,property,constraint,path,nodeshape,greaterEq,assignN,assignE,satisfiesN,satisfiesE,edgeshape,not,targetN,targetE,node,min},
}

\newcommand{\mref}[2]{\hyperref[#2]{#1~\ref*{#2}}}
\definecolor{ultralightgray}{HTML}{eeeeee}

\reportmode 

\begin{document}

\pdfstringdefDisableCommands{%
  \def\\{}%
  \def\texttt#1{}%
}

\newcommand{\progs}{\pshacl: Property Graph Shapes Language}
  
\reportorpaper{
    \title{\progs{}\\(Extended Version)}
    \titlerunning{\progs}
}{
    \title{\progs}
}

\debugonly{
    \title{\progs{}\\\texttt{({\color{reportcolor} REPORT} {\color{papercolor} PAPER})}}
    \titlerunning{\progs}
    \authorrunning{Philipp Seifer, Ralf Lämmel, Steffen Staab}
}

\author{
  \reportorpaper{
    Philipp Seifer\inst{1} \and
    Ralf Lämmel\inst{1} \and
    Steffen Staab\inst{2,3}
  }{
    Philipp Seifer\inst{1}\orcidID{0000-0002-7421-2060} \and
    Ralf Lämmel\inst{1}\orcidID{0000-0001-9946-4363} \and
    Steffen Staab\inst{2,3}\orcidID{0000-0002-0780-4154}
  }
}
\institute{
  The Software Languages Team, University of Koblenz-Landau, Germany
  \email{\{pseifer,laemmel\}@uni-koblenz.de} \and
  Institute for Parallel and Distributed Systems, University of Stuttgart,
  Germany
  \email{steffen.staab@ipvs.uni-stuttgart.de}\and
  Web and Internet Science Research Group, University of Southampton, England
}

\maketitle

\begin{abstract}
Property graphs constitute data models for representing knowledge graphs.
They allow for the convenient representation of facts, including facts about facts, represented by triples in subject or object position of other triples.
Knowledge graphs such as Wikidata are created by a diversity of contributors and a range of sources leaving them prone to two types of errors.
The first type of error, falsity of facts, is addressed by property graphs through the representation of provenance and validity, making triples occur as first-order objects in subject position of metadata triples.
The second type of error, violation of domain constraints, has not been addressed with regard to property graphs so far.
In RDF representations, this error can be addressed by shape languages such as SHACL or ShEx, which allow for checking whether graphs are valid with respect to a set of domain constraints.
Borrowing ideas from the syntax and semantics definitions of SHACL, we design a shape language for property graphs, ProGS, which allows for formulating shape constraints on property graphs including their specific constructs, such as edges with identities and key-value annotations to both nodes and edges.
We define a formal semantics of ProGS, investigate the resulting complexity of validating property graphs against sets of ProGS shapes, compare with corresponding results for SHACL, and implement a prototypical validator that utilizes answer set programming.

\keywords{Property Graphs \and Graph Validation \and SHACL}
\end{abstract}


\section{Introduction}
\label{sec:introduction}

%
Knowledge graphs such as Wikidata~\cite{DBLP:journals/cacm/VrandecicK14} require a data model that allows for the representation of data annotations.
While property graphs serve well as data models for representing such knowledge graphs, they lack sufficient means for validation against domain constraints, for instance required provenance annotations.
%
%
The shapes constraint language SHACL~\cite{shacl} was introduced to allow for validating knowledge graphs that use the RDF data model~\cite{rdf}.
Wikidata and other knowledge graphs, however, make use of triples in subject position to represent provenance metadata, such as references or dates, going beyond the capabilities of the RDF framework.
Similar to extensions of RDF, such as RDF*~\cite{DBLP:conf/semweb/Hartig17} or~aRDF~\cite{DBLP:journals/tocl/UdreaRS10}, property graphs are a promising data model for meeting the modelling needs of annotated knowledge graphs. 
Recent property-graph data model (and query language) proposals include G-CORE~\cite{DBLP:conf/sigmod/AnglesABBFGLPPS18} and the upcoming GQL standard~\cite{gql}, as well as the recently established openCypher standard~\cite{cypher}.
They have attracted a lot of research interest and popularity in practical use-cases~\cite{DBLP:conf/sle/SeiferHLLS19}.
%

Property-graph models differ from RDF in substantial ways, featuring edges with identities (allowing multiple edges between nodes with the same sets of labels) and property annotations (that is key-value annotations) on edges.
A schema or shape-based validation language must account for these differences.
%
%
While there exist efforts to formally define property graph schema languages~\cite{DBLP:conf/grades/HartigH19,DBLP:journals/access/AnglesTT20}, and some practical implementations support simple schemata~\cite{neoschema} (\eg uniqueness constraints) or even enable SHACL validation for RDF compatible subsets of the data graph~\cite{neosemantics}, they do not allow for expressing shape constraints involving 
all elements of property graphs. 
In particular, existing approaches lack support for qualified number restrictions over edge identities, path expressions or the targeted validation of edges.
%
%

Consider the example graph $\examplegraph$ depicting employment relationships in \Cref{fig:intro}.
Some of the nodes and edges have property annotations. 
The edge with identity $200$, for example, has the annotation $\labelname{since}$ with values $\{01/01/1970\}$.
One may wish to define \emph{shapes} to require that all edges labelled \labelname{worksFor} have such metadata annotations.
Shapes that constrain \labelname{Employee} or \labelname{Company} and their interrelationships will lead to recursive descriptions and thus require a corresponding semantics.
Like \cite{DBLP:conf/semweb/CormanRS18}, we adopt a model-based formal semantics based on the notion of (partial) assignments that map  nodes and edges to sets of shape names and constitute the basis for a three-valued evaluation function.

\begin{figure}[t]
    \center
    \setlength{\vgap}{1cm}
\setlength{\hgap}{1.5cm}
\setlength{\kvgap}{0.2cm}

\begin{tikzpicture}

\node[iri] (tim) {\nodelabel{100}{Person Employee}};

\node[literal, below=1\kvgap of tim] (timdata) {
\propertylabelstring{name}{Tim Canterbury}\\
\propertylabelint{age}{30}};

\draw[irim1] (tim) -- (timdata);

\node[iri, right=3\vgap of tim] (wh) {\nodelabel{101}{Company}}
edge[arrin] node[label](wf1){\edgelabel{200}{worksFor}} (tim);

\node[literal, below=1\kvgap of wh] (whdata) {
\propertylabelstring{name}{Wernham Hogg}};

\draw[irim1] (wh) -- (whdata);

\node[literal, above=1\kvgap of wf1, align=left] (wf1data) { \propertylabeldate{since}{01/01/1970}};

\draw[irim1] (wf1) -- (wf1data);

\node[iri, above=1\hgap of tim] (gareth) {\nodelabel{102}{Employee}}
edge[arrin, bend left=20] node[label](gareth1){\edgelabel{201}{colleagueOf}} (tim);

\node[literal, above=1\kvgap of gareth] (garethdata) {
\propertylabelstring{name}{Gareth Keenan}\\
\propertylabelstrings{role}{sales}{team leader}};

\draw[irim1] (gareth) -- (garethdata);

\draw [-] (tim.west) edge[arrin, bend left=60] node[label]{\edgelabel{202}{colleagueOf}} (gareth.west);

\draw [-] (wh) edge[arrin, bend right=20] node[label](wf2){\edgelabel{203}{worksFor}} (gareth.east);

\node[literal, above right=1\kvgap of wf2, align=left] (wf2data) {\propertylabeldate{since}{02/08/2020}};

\draw[irim1] (wf2) -- (wf2data);

\end{tikzpicture}
    \caption{Example property graph $\examplegraph$ showing employment relationships in G-CORE style: Nodes are depicted as rounded boxes.
    Each node has exactly one identifier, \eg $100$ or $101$, and it has zero or more labels, \eg $\{$\labelname{Person}, \labelname{Employee}$\}$ or $\{$\labelname{Company}$\}$.
    Each edge has an identifier, \eg $200$, as well as zero or more labels, \eg $\{$\labelname{worksFor}$\}$.
    Both nodes and edges may have a set of affiliated properties (key-value pairs shown in rectangular boxes), \eg $\{\propertylabelint{age}{30}\}$ or $\{\propertylabeldate{since}{01/01/1970}\}$.}
    \label{fig:intro}
\end{figure}
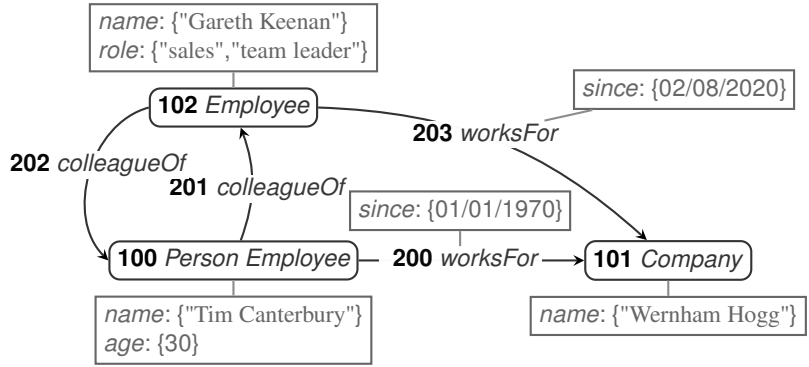
\paragraph{Contributions}
We present \pshacl, a shape language for property graphs that allows for formulating domain constraints and that significantly extends SHACL to property graph data models.
\pshacl comprises property-graph specific features, including shapes for edges with identities, qualified number restrictions over such edges and constraints on  properties and their values.
We define the formal semantics for validating graphs with \pshacl shapes, including cyclic, recursive shape references, based on the notion of partial faithful assignments inspired from~\cite{DBLP:conf/semweb/CormanRS18}.
We analyse the complexity of validating property graphs against sets of \pshacl constraints.
We show that \pshacl validation is NP complete, thus remaining in the same complexity class as SHACL while increasing expressiveness.
We provide a prototypical reference implementation relying on answer set programming, available on GitHub.
\paragraph{Outline}
The remainder of this paper is structured as follows.
\Cref{sec:preliminaries} gives a short overview of property graph models.
In \Cref{sec:pgshapes} we define the abstract syntax and semantics of \pshacl, including assignment-based validation of graphs against a set of \pshacl shapes.
\Cref{sec:complexity} analyses the complexity of the \pshacl graph validation problem.
\Cref{sec:implementation} investigates implementation approaches for \pshacl and introduces a prototypical implementation relying on an encoding of the validation problem as an answer set program.
\Cref{sec:relatedwork} discusses related work and \Cref{sec:summary} concludes the paper.
\section{Foundations}
\label{sec:preliminaries}

Before providing a working definition of property graphs as the basis of \pshacl, we compare existing property graph models to determine essential features.
To this end, consider \Cref{tab:compare}.
We compare the property graph models underlying the graph query languages G-CORE~\cite{DBLP:conf/sigmod/AnglesABBFGLPPS18}, Cypher~\cite{DBLP:conf/sigmod/FrancisGGLLMPRS18}, Gremlin~\cite{gremlin}, and PGQL~\cite{pgql}; we also include the RDF~\cite{rdf} data model and RDF*~\cite{DBLP:conf/semweb/Hartig17} as a point of reference.

We use the example depicted in \Cref{fig:wikidata}, an excerpt from Wikidata, to illustrate the differences between property graphs, RDF and RDF*.
The defining feature of property graphs are properties, that is key-value pairs, on edges and nodes.
For example, $\labelname{point in time}$ in \Cref{fig:wikidata} could be represented as such a property annotation for the edge labelled $\labelname{nominated for}$.
Property keys are string literals, while value domains vary between approaches, ranging from simple scalar values and strings to lists or maps of values.
The key differences to RDF arise from the fact that edges in property graphs have identities.
The edge $\labelname{nominated for}$, for example, would have a unique identity acting as a target for property annotations.
While this is not possible in plain RDF, node properties can be simulated through edges to literal nodes.
RDF* extends RDF by introducing triples that are first-order (FO) objects, meaning they can occur in both subject and object position of other triples.
This importantly subsumes edge properties, again through an encoding of literal nodes.
While not using RDF*, Wikidata also allows for annotations on edges referencing other resources.
This highlights the key difference between FO triples and property annotations:
While $\labelname{point in time}$ could be represented as a property annotation on the $\labelname{nominated for}$ edge, $\labelname{for work}$ could not.

There are some further differences between the various property graph models.
Support for labels differs between approaches, ranging from sets of labels on both nodes and edges (G-CORE, PGQL) to no support for node labels in Gremlin and single edge types in both Gremlin and Cypher.
Finally, only G-CORE features paths as FO objects, \ie paths that can be annotated with property annotations and labels.

\begin{table}
\caption{Comparison of feature support for common property graph models and RDF.}
\begin{center}
\begin{tabular}{ l c c c c c c }
\toprule
 & G-CORE & Cypher & Gremlin & PGQL & RDF* & RDF \\
 \midrule
 Nodes as FO objects & + & + & + & + & + & +\\
 Node properties & + & + & + & + & literals & literals \\  
 Node labels & set & set & none & set & set (rdf:type)  & set (rdf:type) \\  
 Triples/Edges as FO objects & + & + & + & + & + & -\\
 Edge properties & + & + & + & + & + & - \\  
 Edge labels  & set & single & single & set & single & single \\  
 Paths as FO objects & + & - & - & - & - & -\\
 Path properties & + & - & - & - & - & - \\  
 Path labels  & set & - & - & - & - & - \\
 \bottomrule
\end{tabular}
\end{center}
\label{tab:compare}
\end{table}

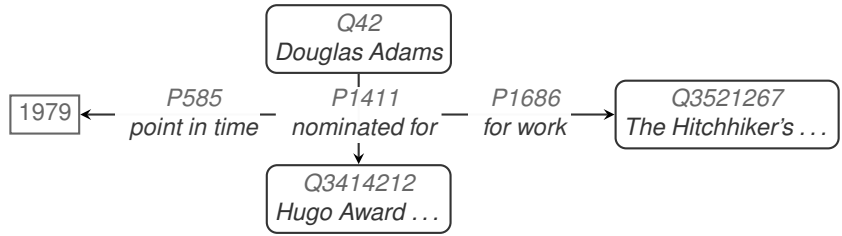
\begin{figure}[t]
    \center
    \setlength{\vgap}{1.2cm}
\setlength{\hgap}{1cm}
\setlength{\kvgap}{0.2cm}

\begin{tikzpicture}

\node[iri] (q42) {\wikilabel{Q42}{Douglas Adams}};

\node[iri, below=1\vgap of q42,align=center] (q3414212) {\wikilabel{Q3414212}{Hugo Award \ldots}} edge[arrin] node[label,align=center](p1411){\wikilabel{P1411}{nominated for}} (q42);

\node[literal,left=2.6\hgap of p1411,align=center] {1979} edge[arrin] node[label,align=center](p585){\wikilabel{P585}{point in time}} (p1411);

\node[iri,right=2.3\hgap of p1411,align=center] {\wikilabel{Q3521267}{The Hitchhiker's \ldots}} edge[arrin] node[label,align=center](p1686){\wikilabel{P1686}{for work}} (p1411);

\end{tikzpicture}
    \caption{Excerpt from Wikidata.}
    \label{fig:wikidata}
\end{figure}
\subsection{Definition of Property Graphs}
The formalization of the property graph model we use as a basis for the definition of \pshacl is based on the data model presented for G-CORE~\cite{DBLP:conf/sigmod/AnglesABBFGLPPS18}.
We do not consider first-class paths, and instead restrict the model to the core subset shared with other property graph models as discussed in the previous section.
In terms of value domains in properties, we provide exemplary support for the types \texttt{string}, \texttt{int} and \texttt{date}, without loss of generality.

Let the set of labels $L=L_N \cup L_E$ where $L_N$ is an infinite set of node labels and $L_E$ an infinite set of edge labels. 
As a matter of convention, we use \labelname{CamelCase} for all $l_N \in L_N$ and \labelname{camelCase} for all $l_E \in L_E$.
Let $K$ be an infinite set of property names (or keys) and $V$ an infinite set of literal values from the union of sets in $T \in \{ \texttt{int}, \texttt{string}, \texttt{date}\}$.
We refer to elements of $T$ as the type of the respective value.
Let furthermore $\textrm{FSET}(X)$ denote all finite subsets of a set $X$.

\begin{definition}[Property Graph]
  A property graph is a tuple $G=(N,E,\rho,\lambda,\sigma)$, where
  $N$ denotes a set of node identifiers and $E$ a set of edge identifiers, with $N \cap E = \emptyset$,
  $\rho : E \rightarrow (N \times N)$ is a total function,
  $\lambda : (N \cup E) \rightarrow \textrm{FSET}(L)$ is a total function,
  $\sigma : (N \cup E)\times K\rightarrow \textrm{FSET}(V)$ is a total function for which a finite set of tuples $(x,k)\in(N\cup E) \times K$ exists such that $\sigma(x,k)\ne\emptyset$.
\end{definition}

\noindent
A property graph consists of a set of nodes $n\in N$ and edges $e\in E$, where $\rho$ maps elements of $E$ to pairs of nodes.
The function $\lambda$ maps nodes and edges to all assigned labels $l\in L$ and likewise the function $\sigma$ maps pairs of nodes or edges, and property names to the property values assigned to them.
The example in \Cref{fig:ppg1} shows the property graph visualized in \Cref{fig:intro} using the formal definition.
Note, that we omit finitely many elements of the domain of $\lambda$ that are mapped to $\emptyset$ (none in this particular example) and infinitely many elements of the domain of $\sigma$ that are mapped to $\emptyset$.
\begin{figure}[ht]
    \centering
    \input{2/figures/example-graph-formal}
    \caption{Formal model for the example property graph $\examplegraph$ rendered in \Cref{fig:intro}.}
    \label{fig:ppg1}
\end{figure}
\section{Shapes for Property Graphs}
\label{sec:pgshapes}

Our shape language for property graph validation, called \pshacl, has been inspired by SHACL~\cite{shacl}, the W3C recommendation for writing and evaluating RDF graph validation constraints.
More specifically, we base the \pshacl shape language on the abstract syntax proposed by \cite{DBLP:conf/semweb/CormanRS18}, which formalizes a syntactic core of SHACL.
Corman et al.\ \cite{DBLP:conf/semweb/CormanRS18} also defined a formal semantics for this syntactic core that addresses recursion, in particular.
We facilitate the understanding of differences between SHACL and \pshacl by colour coding. Expressions that we borrow from SHACL will be displayed in  \shacltranslate{black font}, while novel expressions will be  coded
in \shaclextend{blue font}.

\subsection{Requirements on a Property Graph Shapes Language}
Requirements for our target language stem from the differences between the RDF and property graph data models, which we mentioned in \Cref{sec:preliminaries}.
\Cref{tab:mapping} explains how RDF may be mapped to the G-CORE property graph model.
Based on this mapping we design \pshacl to adopt language constructs from SHACL. The reader may note that this mapping includes some design decisions that are not unique, \eg we interpret class instantiations as corresponding to G-CORE labellings of nodes.
We follow a simplification of the third mapping $\mathcal{I}\mathcal{M}_3$ discussed in~\cite{DBLP:journals/access/AnglesTT20}, \eg by excluding blank nodes.

\begin{table}[t]
    \caption{Sketching correspondences between the RDF and G-CORE graph models.}
    \centering
    \begin{tabular}{p{5cm}p{3cm}l}\toprule
    Description & RDF  & G-CORE\,/\,\pshacl  \\ \midrule

      Node id  $i$  & IRI  $i$ & $i\in N$ \\
      Node $n$ has label $l$ & $n$ rdf:type $l$.  & $l\in \lambda(n)$\\
      Node $n$ has key $k$ with value $v$ &  $n\, k\, v.$ & $v\in \sigma(n,k)$\\
      Edge id $i$   & not available  & $i \in E$ \\
      Edge label $l$, in triple $s\, p\, o.$ & $s\, l\, o$ & $l\in \lambda(p)$\\
      Edge $e$ has key $k$ with value $v$ & not available & $v\in\sigma(e,k)$\\ 
            Triple $s\, p\, o.$ & $s\, p\, o.$ & $p\in \lambda(i), \rho(i)=(s,o)$\\
      \bottomrule
    \end{tabular}
    \vspace{0.3cm}
    \label{tab:mapping}
\end{table}

Edges in property graphs have identities, necessitating two distinct kinds of shapes for nodes (R1) and for edges (R2), as well as two kinds of qualified number restrictions for nodes, counting edges (R3) and counting reachable nodes via some path (R4).
Property annotations require dedicated constraints dealing with the set of values reachable via a specific key, for both nodes (R5) and edges (R6).
The presence of properties must also be considered for constraints that include comparison operations (R7).
Lastly, the existence of certain properties, or properties with certain values, also require new means of targeting nodes and edges in target queries (R8).
\subsection{Definition of Shapes}

Intuitively, a shape defines constraints on how certain nodes or edges in a graph are formed.
As both nodes and edges in property graphs have identities, we define \emph{node shapes} that apply to nodes and \emph{edge shapes} that apply to edges.
Each shape is a triple consisting of a shape name, a constraint, and a target query defining which nodes or which edges of a graph must conform to the shape, \ie fulfil all of its constraints, for the graph to be considered in conformance with the shape.
\begin{example}
  The node shape $\nodeshape[\shapename{PersonShape}][\labelname{Employee}][\labelname{Person}]$ is a triple with the shape name $\shapename{PersonShape}$, the constraint $\labelname{Employee}$, which requires that each graph node assigned this shape has the label $\labelname{Employee}$, and the target query $\labelname{Person}$, meaning all nodes with the label $\labelname{Person}$ are targets of this shape.
  For the graph $\examplegraph$ in Figure~\ref{fig:intro}, node $100$ conforms this shape, whereas node $102$ does not, lacking the $\labelname{Person}$ label. Given that at least one target node does not conform to the constraint, the entire graph does not conform to $\shapename{PersonShape}$.
\end{example}
As show in the first example, we use $\nodeshape$ to indicate triples that are node shapes and use $\edgeshape$ to refer to triples that are edge shapes.

Before introducing their components, we define path expressions $p$ in \cref{eq:qn} in analogy to property path expressions defined in SHACL~\cite{shacl}, which are in turn based on path expressions in the SPARQL query language.
A path expression, when evaluated on a starting node, describes the set of nodes reachable from this node via paths that match the path expression.
\begin{equation}
  \shacltranslate{
    p ::= \textrm{ } 
    l_E \mid p^- \mid p_1/p_2 \mid p_1 || p_2 \mid p* \mid p+ \mid \,?p
  }
  \label{eq:qn}
\end{equation}
Path expressions may include edge labels $l_E$, inverse paths $p^-$, path sequences $p_1/p_2$, alternate paths $p_1 || p_2$ and zero or more ($p*$), one or more ($p+$) and zero or one ($?p$) expressions.
Note the minor difference to paths in RDF graphs, in that edges in property graphs may have multiple labels.
%
\begin{example}
  The path $\labelname{worksFor}/\labelname{worksFor}^-$ describes the set of all colleagues of a starting node $n$ (including $n$ itself), by first finding all employers of $n$ (\ie nodes reachable from $n$ via an edge with label $\labelname{worksFor}$) and then all employees of those employers (\ie nodes with incoming $\labelname{worksFor}$ edges).
  For the graph $\examplegraph$ in \Cref{fig:intro} and starting node $100$, the result of evaluating this path would be the same as evaluating $\labelname{colleagueOf}*$, namely the set $\{100,102\}$.
\end{example}

Let the set of shapes $S = S_N \cup S_E$ consist of node and edge shapes and the set of shape names be called Names($S$). 
A node shape is a tuple $\nodeshape$ consisting of a shape name $s_N \in \textrm{Names}(S_N)$, a node constraint $\phi_N$ and a query for target nodes $q_N$.
A query for target nodes is either $\bot$, meaning the query has no targets, an explicitly targeted node $n \in N$, all nodes with  label $l_N \in L_N$, all nodes with property $k \in K$ or possibly further constrained as $\targetvk{v}{k}$ by a concrete value $v \in V$.
The syntax of target node queries $q_N$ is summarized in \cref{eq:tn}.
We write $\evalq{q_N}$ for the evaluation of a target node query, which is defined in \Cref{fig:nodetargetqueries}.
\begin{figure}[t]
  \begin{align*}
    \evalq{\bot} &= \emptyset\\
    \evalq{n} &= \{n\}\\
    \evalq{l_N} &= \{ n \mid n \in N \land l_N \in \lambda(n) \}\\
    \evalq{k} &= \{ n \mid n \in N \land \sigma(n,k) \neq \emptyset\}\\
    \evalq{\targetvk{v}{k}} &= \{n \mid n \in N \land v \in \sigma(n,k) \}
  \end{align*}
  \caption{Evaluation of target node queries.}
  \label{fig:nodetargetqueries}
\end{figure}
\begin{equation}
  \label{eq:tn}
  \begin{split}
    q_N ::= \textrm{ } \shacltranslate{\bot} 
    \mid \shacltranslate{n}
    \mid \shacltranslate{l_N} 
    \mid k 
    \mid \targetvk{v}{k}
  \end{split}
\end{equation}
\begin{example}
  The target query $q_N = \labelname{Employee}$ targets all nodes that are labelled with the label $\labelname{Employee}$.
  The set of targets when evaluating $q_N$ on the example graph $\examplegraph$ in \Cref{fig:intro} is therefore $\evalq[\examplegraph]{q_N} = \{100,102\}$.
\end{example}

Node constraints $\phi_N$ essentially specify which outgoing or incoming edges, which labels, or which properties a targeted node must have.
Assuming $s_N \in S_N$, $n \in N$, $l_N\in L_N$, $k\in K$, $i\in\mathbb{N}$, comparison operations $\odot$ for sets or singleton sets (\eg $=$, $<$, $\subset$) and arbitrary value predicate functions $f : V \rightarrow \{0,1\}$ such as $\geq 0$, $\ne 19$, or type restrictions for a specific data type such as \texttt{int}, \texttt{string} or \texttt{date}, node constraints $\phi_N$ are defined as in \cref{eq:phin}.
\begin{equation}
  \label{eq:phin}
  \begin{split}
    \phi_N ::= 
    \textrm{ } \shacltranslate{\top}\ & 
    \mid \shacltranslate{s_N}
    \mid \shacltranslate{n}
    \mid \shacltranslate{l_N}
    \mid \shacltranslate{\neg\phi_N}
    \mid \shacltranslate{\phi_N^1 \wedge \phi_N^2}
    \mid\, \shacltranslate{\geqslant_i p.\phi_N}
    \mid\, \shacltranslate{\odot\, (p_1, p_2)}\\
    &\mid\, \shaclextend{\geqslant_i k.f }
    \mid\, \shaclextend{\geqleft{i}{\phi_E}}
    \mid\, \shaclextend{\geqright{i}{\phi_E}}
    \mid\, \shaclextend{\odot\, (p_1, k_1, p_2, k_2)}
    \mid\, \shaclextend{\odot\, (k_1, k_2)}
  \end{split}
\end{equation}
A node constraint may be always satisfied ($\top$), reference another node shape with name $s_N$ that must be satisfied, require a specific node identity $n$ in this place or require a node label $l_N$. 
It may also be the negation $\neg\phi_N$ or conjunction $\phi_N^1 \wedge \phi_N^2$ of other node constraints.
Furthermore, the constraint $\geqslant_i p.\phi_N$ requires $i$ nodes that can be reached via path $p$ to conform to $\phi_N$.
$\odot\,(p_1,p_2)$ is an arbitrary comparison operation between sets of node identities that can be reached via the two path expressions $p_1$ and $p_2$.
\begin{example}
  Consider the shape $\nodeshape[s_1][\geqslant_1 \labelname{colleagueOf}.\labelname{Person}][\labelname{Employee}]$ targeting all nodes with the label $\labelname{Employee}$.
  $s_1$ requires at least one path $\labelname{colleagueOf}$, \ie an outgoing edge that has the label $\labelname{colleagueOf}$, to a node which has the label $\labelname{Person}$. 
  For the graph in \Cref{fig:intro}, node $102$ satisfies this constraint, because the only node reachable via path $\labelname{colleagueOf}$ is node $100$, and $\labelname{Person} \in \lambda(100)$.
  With analogous reasoning, the constraint does not hold for node $100$, because $\labelname{Person} \not\in \lambda(102)$
\label{ex:colleague}
\end{example}
The aforementioned constraints were essentially transferred from core constraint components of the SHACL language.
Novel kinds of constraints are printed in \shaclextend{blue font}.
A qualified number restriction ${\geqslant_i k.f}$ restricts the number of values matching the predicate $f$ for the property $k$.
The qualified number constraints $\geqleft{i}{\phi_E}$ and $\geqright{i}{\phi_E}$ require $i$ incoming or outgoing edges that conform to the given edge constraint $\phi_E$ (defined below).
$\odot\,(p_1,k_1,p_2,k_2)$ compares the annotated sets of values for properties $k_1$ and $k_2$, reached via paths $p_1$ and $p_2$ and $\odot\,(k_1,k_2)$ does the same for the current node.
\begin{example}
  Consider the shape $\nodeshape[s_2][\geqslant_2 \propname{role}.\texttt{string}\ \wedge\ s_1][\targetvk{\stringvalue{Gareth Keenan}}{\propname{name}}]$, which targets all nodes $n$ where $\stringvalue{Gareth Keenan} \in \sigma(n, \propname{name})$.
  For the graph $\examplegraph$ in \Cref{fig:intro}, node $102$ is the only target.
  The constraint $\geqslant_2 \propname{role}.\texttt{string}\ \wedge\ s_1$ requires that this node conform to shape $s_1$ from \Cref{ex:colleague}, as well as that the $\propname{role}$ property has at least two elements of type $\texttt{string}$.
  From \Cref{ex:colleague} it follows that $102$ conforms to $s_1$.
  The property $\sigma(102, \propname{role})$ has two values $\{\stringvalue{sales}, \stringvalue{team leader}\}$, both of which are strings.
  Therefore, node $102$ conforms to $s_2$. 
  Since node $102$ is the only target of $s_2$, $\examplegraph$ conforms to $s_2$ as well.
  \label{ex:and}
\end{example}
Edge shapes apply to edges and, just as a node shape, require specific labels or properties for all targeted edges. 
Similarly to how node shapes constrain outgoing and incoming edges, edge shapes may constrain the source or destination node of an edge.

An edge shape is a tuple $\edgeshape$ consisting of shape name $s_E\in \textrm{Names}(S_E)$, an edge constraint $\phi_E$ and a target edge query $q_E$.
Edge target queries are defined analogously to node target queries in \cref{eq:et} and \Cref{fig:edgetargetqueries}.
\begin{equation}
    \shaclextend{q_E} ::= 
    \textrm{ } \shaclextend{\bot}
    \mid \shaclextend{e}
    \mid \shaclextend{l_E}
    \mid \shaclextend{k}
    \mid \shaclextend{\targetvk{v}{k}}
  \label{eq:et}
\end{equation}
\begin{figure}[t]
  \begin{align*}
    \evalq{\bot} &= \emptyset\\
    \evalq{e} &= \{e\}\\
    \evalq{l_E} &= \{ e \mid e \in E \land l_E \in \lambda(e) \}\\
    \evalq{k} &= \{ e \mid e \in E \land \sigma(e,k) \neq \emptyset\}\\
    \evalq{\targetvk{v}{k}} &= \{e \mid e \in E \land v \in \sigma(e,k) \}
  \end{align*}
  \caption{Evaluation of target edge queries.}
  \label{fig:edgetargetqueries}
\end{figure}

Most constraint components of edge constraints $\phi_E$ are defined similarly to node constraints $\phi_N$, albeit in terms of the respective edge identities $e$, edge labels $l_E$ and edge shapes $s_E$.
Unique to edge constraints are the constraints $\Leftarrow \phi_N$ and $\Rightarrow \phi_N$, which constrain source or destination nodes of an edge to conform to a node shape $\phi_N$.
Edge constraints $\phi_E$ are defined as in \cref{eq:phie}.
\begin{equation}
  \label{eq:phie}
  \begin{split}
    \shaclextend{\phi_E} ::= 
    \textrm{ } \shaclextend{\top}
    \mid \shaclextend{s_E}
    \mid \shaclextend{e}
    \mid \shaclextend{l_E}
    \mid \shaclextend{\neg\phi_E}
    \mid \shaclextend{\phi_E^1 \wedge \phi_E^2}
    \mid \shaclextend{\geqslant_i k.f}
    \mid \shaclextend{\Rightarrow \phi_N}
    \mid \shaclextend{\Leftarrow \phi_N}
    \mid \shaclextend{\odot\, (k_1, k_2)}
  \end{split}
\end{equation}
\begin{example}
Consider $\edgeshape[s_3][\Leftarrow \labelname{Person}\ \wedge \geqslant_1 \labelname{since}.(\geq 01/01/2020)][\labelname{worksFor}]$ which targets edges with the label $\labelname{worksFor}$.
For the two matching edges of graph $\examplegraph$ in \Cref{fig:intro}, $200$ and $203$, only $200$ fulfils the constraint $\Leftarrow \labelname{Person}$, since $\labelname{Person} \in \lambda(100)$ and $\rho(200) = (100,101)$.
That is, the source node of edge $200$ has the label $\labelname{Person}$.
Only edge $203$ fulfils the constraint $\geqslant_1 \labelname{since}.\geq 01/01/2020$, because at least one element of $\sigma(203,\labelname{since}) = \{02/08/2020\}$ fulfil the given value predicate $\geq 01/01/2020$, because $02/08/2020 \geq 01/01/2020$.
Neither edge fulfils $s_3$.
\end{example}
\begin{example}
There is a difference between a node constraint $\geqslant_3 \labelname{colleagueOf}.\labelname{Person}$ and a node constraint $\geqright{3}{(\labelname{colleagueOf}\ \wedge \Rightarrow \labelname{Person})}$.
In the first case, we require $3$ distinct nodes with the label $\labelname{Person}$, reachable via edges that match $\labelname{colleagueOf}$.
In the second case, we require $3$ outgoing edges labelled $\labelname{colleagueOf}$ with destination nodes labelled $\labelname{Person}$. 
The nodes in the second case are not required to be distinct.
Indeed, a graph with a single node having three self-loops could potentially fulfil the second, but never the first constraint.
\end{example}
In addition to these core constraints, we define useful syntactic sugar for both node constraints $\phi_N$ and edge constraint $\phi_E$ as shown in \Cref{fig:sugar}.
For target queries, both conjunction and disjunction can also be defined as syntactic sugar (we use $\phi$ and $q$ to mean either a node or edge constraint and query, respectively).
Any shape with target $q_1 \wedge q_2$ and constraint $\phi$ is equivalent to a shape with target $q_1$ and the constraint $(\phi \wedge \phi_{q_2}) \vee \neg \phi_{q_2}$, where $\phi_{q_2}$ is the constraint equivalent to the target query (\ie validating exactly the targets).
Any shape $s$ with target $q_1 \vee q_2$ and constraint $\phi$ can be expressed via two utility shapes with target $q_1$ and constraint $s$ and target $q_2$ and constraint $s$, as well as the shape $s$ with target $\bot$ and constraint $\phi$.

\begin{figure}[t]
\begin{center}
\begin{equation*}
  \begin{split}
    \bot &:= \neg \top\\
    \leqleft{i}{\phi_E} &:= \neg\geqleft{i+1}{\phi_E}\\
    \leqslant_i p.\phi_N &:= \neg\geqslant_{i+1} p.\phi_N\\
    \leqslant_i k.f &:= \neg\geqslant_{i+1} k.f\\
    \eqleft{i}{\phi_E} &:=\, \geqleft{i}{\phi_E} \,\wedge\, \leqleft{i}{\phi_E}\\
    =_i p.\phi_N &:=\, \geqslant_i p.\phi_N \,\wedge\, \leqslant_i p.\phi_N\\
    =_i k.f &:=\, \geqslant_i k.f \,\wedge\, \leqslant_i k.f\\
  \end{split}
  \hspace{1cm}
  \begin{split}
    \existsleft{\phi_E} &:=\, \geqleft{1}{\phi_E}\\
    \exists p.\phi_N &:=\, \geqslant_1 p.\phi_N\\
    \exists k.f &:=\, \geqslant_1 k.f\\
    \forallleft{\phi_E} &:=\, \leqleft{0}{\neg \phi_E}\\
    \forall p.\phi_N &:=\, \leqslant_0 p.\neg \phi_N\\
    \forall k.f &:=\, \leqslant_0 k.\neg f\\
    \phi_1 \vee \phi_2 &:= \neg (\neg \phi_1 \wedge \neg \phi_2)
  \end{split}
\end{equation*}
\end{center}
\caption{Syntactic sugar for constraints, where $\phi$ is placeholder for either $\phi_N$ or $\phi_E$.
Definitions for syntactic sugar related to $\geqright{i}{\phi_E}$ are omitted, since they are analogous to $\geqleft{i}{\phi_E}$.}
\label{fig:sugar}
\end{figure}
\subsection{Shape Semantics}
\label{subsec:eval}
Our definition of \pshacl allows shape names to occur in constraints, meaning recursive cycles of references to other shapes can arise.
Therefore, we follow an approach defined for recursive SHACL~\cite{DBLP:conf/semweb/CormanRS18} and define evaluation of shapes on the basis of \emph{partial assignments} for graph nodes and edges to sets of shapes.
Our approach then relies on validating a given assignment in polynomial time (\eg by guessing an assignment).

We formally define assignments on the basis of atoms, such that for each atom that pairs the name of a node shape with a node $s_N(n)$ or the name of an edge shape with an edge $s_E(e)$  a truth value from  $\{0,0.5,1\}$ may be assigned.
\begin{definition}[Atoms]
  For a property graph $G=(N,E,\rho,\lambda,\sigma)$ and a set of shapes $S = S_N \cup S_E$, the set {\normalfont $\textrm{atoms}(G,S) = \textrm{atoms}_N(G,S_N)\ \cup\ \textrm{atoms}_E(G,S_E)$} where 
  {\normalfont $\textrm{atoms}_N(G,S_N) = \{ s_N(n) \mid s_N \in S_N \land n \in N\}$} and
  {\normalfont $\textrm{atoms}_E(G,S_N) = \{ s_E(e) \mid s_E \in S_E \land e \in E\}$} 
  is called the set of atoms of $G$ and $S$.
\end{definition}
For the set of atoms of $G$ and $S$, meaning essentially all tuples of shapes in $S$ and nodes (or edges, respectively) in $G$, we define a partial assignment as a function $\Sigma$ that maps for $x \in N \cup E$ all atoms $s(x) \in \textrm{atoms}(G,S)$ to $1$, if the shape $s$ is assigned to $x$, to $0$ if $\neg s$ is assigned to $x$, and to $0.5$ otherwise.
\begin{definition}[Partial Assignment]
  Let $G$ be a property graph and $S$ a set of shapes. 
  A partial assignment \shapeassignment is a total function {\normalfont $\shapeassignment : \textrm{atoms}(G,S) \rightarrow \{0,0.5,1\}$}.
\end{definition}
Evaluating whether a node $n\in N$ of $G$ satisfies a
constraint $\phi_N$, written $\evaln{\phi_N}$ is defined in \Cref{fig:nconstrainteval} and evaluating whether an edge $e\in E$ of $G$ satisfies a constraint $\phi_E$, written $\evale{\phi_E}$, is defined in \Cref{fig:econstrainteval}. 
In the latter figure we omit cases that are trivially analogous to node shapes.
In both figures, $[P]$ is similar to the Iverson bracket, such that $[P]$ evaluates to \true (the constraint is satisfied) if $P$ is true and \false (the constraint is not satisfied) if $P$ is false.
Conditions for evaluation to $0.5$ are given explicitly.

\begin{figure}[thp]
  \begin{align*}
    \evaln{\top} &= \true \\
    \evaln{s_N} &= \shapeassignment(s_N(n))\\
    \evaln{n'} &= \iverson{n' = n}\\
    \evaln{l_N} & = \iverson{l_N \in \lambda(n)}\\
    \evaln{\neg \phi_N} &= 1 - \evaln{\phi_N}\\
    \evaln{\phi_N^1 \land \phi_N^2} &= \textrm{min} \{ \evaln{\phi_N^1}, \evaln{\phi_N^2}\}\\
    \evaln{\geqslant_i p.\phi_N} &= 
    \begin{cases}
      1 & \vert\{n_2 \mid n_2 \in \evalpath{p} \land \evaln[\shapeassignment][n_2]{\phi_N} = 1 \}\vert \geq i\\
      0 & \vert\evalpath{p}\vert\ - \\
        & \vert \{n_2 \mid n_2 \in \evalpath{p} \land \evaln[\shapeassignment][n_2]{\phi_N} = 0 \}\vert < i\\
      0.5 & \textrm{otherwise}\\
    \end{cases}\\
    \evaln{\odot\,(p_1,p_2)} &= \iverson{\evalpath{p_1} \odot \evalpath{p_2}}\\
    \shaclextend{\evaln{\geqleft{i}{\phi_E}}} &\shaclextend{\ =\ }
    \shaclextend{\begin{cases}
      1 & \vert\{e \mid e \in E \land n_2\in N \land \rho(e) = (n_2, n)\\ & \quad \land \evale{\phi_E} = 1 \}\vert \geq i\\
      0 & \vert\{e \mid e \in E \land n_2\in N \land \rho(e) = (n_2, n) \}\vert\ - \\
        & \vert\{e \mid e \in E \land n_2\in N \land \rho(e) = (n_2, n)\\ & \quad \land \evale{\phi_E} = 0 \}\vert < i\\
      0.5 & \textrm{otherwise}\\
    \end{cases}}\\
    \shaclextend{\evaln{\geqright{i}{\phi_E}}} &\shaclextend{\ =\ }
    \shaclextend{\begin{cases}
      1 & \vert\{e \mid e \in E \land n_2\in N \land \rho(e) = (n, n_2)\\ & \quad \land \evale{\phi_E} = 1 \}\vert \geq i\\
      0 & \vert\{e \mid e \in E \land n_2\in N \land \rho(e) = (n, n_2) \}\vert\ - \\
        & \vert\{e \mid e \in E \land n_2\in N \land \rho(e) = (n, n_2)\\ & \quad \land \evale{\phi_E} = 0 \}\vert < i\\
      0.5 & \textrm{otherwise}\\
    \end{cases}}\\
    \shaclextend{\evaln{\geqslant_i k.f}} &\shaclextend{\ = \iverson{\vert\{v \mid v \in \sigma(n,k) \land f(v)\}\vert \geq i}}\\
    \shaclextend{\evaln{\odot\,(p_1,k_1,p_2,k_2)}} &\shaclextend{\ =
    \begin{aligned}
      & \iversonleft \{v \mid n \in \evalpath{p_1}, v \in \sigma(n,k_1)\}\\
      & \ \odot \{v \mid n \in \evalpath{p_2}, v \in \sigma(n,k_2)\}  \iversonright
    \end{aligned}
    }\\
    \shaclextend{\evaln{\odot\,(k_1,k_2)}} &\shaclextend{\ = \iverson{\sigma(n,k_1) \odot \sigma(n,k_2)}}\\
  \end{align*}
  \caption{Evaluation rules for node constraints over graph $G$ with assignment $\shapeassignment$.}
  \label{fig:nconstrainteval}
\end{figure}

\begin{figure}[thp]
  \begin{align*}
    \shaclextend{\evale{s_E}} & \shaclextend{\ = \shapeassignment(s_E(e))}\\
    \shaclextend{\evale{e'}} & \shaclextend{\ = \iverson{e' = e}}\\
    \shaclextend{\evale{l_E}} & \shaclextend{\ = \iverson{l_E \in \lambda(e)}}\\
    \shaclextend{\evale{\Rightarrow \phi_N}} & \shaclextend{\ = \evaln[\shapeassignment][n_2]{\phi_N} \textrm{ where } (n_1,n_2) = \rho(e)}\\
    \shaclextend{\evale{\Leftarrow \phi_N}} & \shaclextend{\ = \evaln[\shapeassignment][n_1]{\phi_N} \textrm{ where } (n_1,n_2) = \rho(e)}\\
  \end{align*}
  \caption{Evaluation rules for edge constraints over graph $G$ with assignment $\shapeassignment$ (omitting some cases that are analogous to cases in \Cref{fig:nconstrainteval}).}
  \label{fig:econstrainteval}
\end{figure}

The semantics of path expressions are defined in \Cref{fig:pathsemantics}.
We write $\{n_1,\ldots,n_i\} = \evalpath{p}$ for the evaluation of path $p$ on graph $G$, such that nodes $n_1,\ldots,n_i$ can be reached via $p$ from node $n$.

\begin{figure}[th]
  \begin{align*}
    \shaclextend{\evalpath{l_E}} &\shaclextend{\ = \{ n_1 \mid e \in E \land (n,n_1) = \rho(e) \land l_E  \in \lambda(e)\}}\\
    \evalpath{p^-} &= \{n_2 \mid n \in \evalpath[\Sigma][n_2]{p}\} \\
    \evalpath{p_1/p_2} &= \bigcup \{ \evalpath[\shapeassignment][n_1]{p_2} \mid n_1 \in \evalpath{p_1} \} \\
    \evalpath{p_1 || p_2} &= \evalpath{p_1} \cup \evalpath{p_2} \\
    \evalpath{p+} &= 
      \begin{cases}
        \emptyset, & \text{if}\  \evalpath{p} = \emptyset\\
        \evalpath{p} \cup \evalpath{p/p+}, & \text{otherwise}
      \end{cases}\\
    \evalpath{p*} &= \{n\} \cup \evalpath{p+} \\
    \evalpath{?p} &= \{n\} \cup \evalpath{p} \\
  \end{align*}
  \vspace{-1.0cm}
  \caption{Evaluation of path expressions.}
  \label{fig:pathsemantics}
\vspace{-0.2cm}
\end{figure}
In order for a property graph $G$ to be valid with respect to a set of shapes $S$, an assignment must exists which complies with all targets and constraints in $S$. 
Transferring terminology from~\cite{DBLP:conf/semweb/CormanRS18} we call such an assignment \emph{strictly faithful}.
\begin{definition}[Strictly Faithful Assignment]\label{def:faithful_assignment1}
  An assignment $\shapeassignment$ for a property graph $G = (N,E,\rho,\lambda,\sigma)$ and a set of shapes $S$ is strictly faithful, if and only if the following 4 properties hold (given shapes of the form $\nodeshape$ and $\edgeshape$):

  \begin{enumerate}
    \item $\forall\ s_N(n) \in \textrm{atoms}(G,S) : \Sigma(s_N(n)) = \evaln{\phi_N}$
    \item $\forall\ s_E(e) \in \textrm{atoms}(G,S) : \Sigma(s_E(e)) = \evale{\phi_E}$

    \item $\forall n \in \evalq{q_N} : \Sigma(s_N(n)) = 1$
    \item $\forall e \in \evalq{q_E} : \Sigma(s_E(e)) = 1$
  \end{enumerate}
\end{definition}
This means a strictly faithful assignment is an assignment, where all atoms are assigned exactly the result of constraint evaluation, all targets $n \in \evalq{q_N}$ are assigned the respective shape $s_N$, and all targets $e \in \evalq{q_E}$ are assigned the respective shape $s_E$.
We define conformance of a graph with respect to a set of shapes on the basis of faithful assignments.

\begin{definition}[Conformance]\label{def:conformance}
  A property graph $G = (N,E,\rho,\lambda,\sigma)$ conforms to a set of shapes $S$ if and only if there exists at least one assignment $\shapeassignment$ for $G$ and $S$ that is strictly faithful.
\end{definition}
\subsection{Requirements and Relationship to SHACL}

As visualized by the colour coding of our definitions, the syntax of \pshacl is an extension of the $\mathcal{L}$ language formalization of SHACL~\cite{DBLP:conf/semweb/CormanRS18}.
There are some exceptions arising from the existence of edges that have identities in property graphs.
In fulfilment of requirements R3 and R4, \pshacl allows qualifying the number of outgoing and incoming edges as well as reachable nodes, whereas SHACL only needs to be concerned with reachable nodes via some path.

Node shapes in SHACL may target all subjects or objects of an RDF property via \texttt{targetSubjectsOf} and \texttt{targetObjectsOf} expressions.
In \pshacl, these target queries are not required.
Instead, fulfilling requirements R1 and R2, as well as R8, \pshacl allows targeting of edges directly with specialized edge shapes.
The respective source and destination nodes can then be constrained in these shapes via $\Leftarrow \phi_N$ and $\Rightarrow \phi_n$, respectively.

Finally, the handling of RDF literals in SHACL differs from constraints dealing with property annotations on nodes (R5 and R7) in \pshacl, as previously discussed.
In addition, \pshacl allows validating property annotations on edges (R6), which do not exist in RDF.




\section{Complexity}
\label{sec:complexity}

We analyse the complexity of validating a property graph against a set of \pshacl shapes.
Before we define the validation problem $\validation$ through the notion of faithfulness of assignments, we simplify the definition of faithful assignments with respect to target queries, by showing that it suffices to consider only cases where there is exactly one target node.
\begin{proposition}
  \label{prop:singletarget}
  For a graph $G=(N,E,\rho,\lambda,\sigma)$ and a set of shapes $S=S_N \cup S_E$ with target nodes $n \in N$ for each $s_N \in S_N$ and target edges $e \in E$ for each $s_E \in S_E$, a graph $G'$ and set of shapes $S'$ can be constructed in linear time, such that $G$ is valid against $S$ if and only if $G'$ is valid against $S'$ and $S'$ has a single target in $G'$.
\end{proposition}
\begin{proof}[Sketch] 
  Essentially, we construct edges from a new, single target node to previous target nodes and source nodes of target edges.
  Then we adapt constraints appropriately.
  Let $s_N^1, \ldots, s_N^n$ and $s_E^1,\ldots,s_E^n$ be shapes in S with targets $n_1^1$, $\ldots$,$n_1^m$, $\ldots$,$n_n^1$, $\ldots,n_n^m$ and targets $e_1^1$, $\ldots$, $e_1^m$, $\ldots$, $e_n^1$, $\ldots,e_n^m$.
  Extend $G$ with a fresh node $n_0$ and fresh edges $ne_i^j$ with $\rho(ne_i^j) = (n_0, n_i^j)$ for each target $n_i^j$ as well as edges $ee_i^j$ with $\rho(ee_i^j) = (n_0, n_1)$ where $(n_1, n_2) = \rho(e_i^j)$ for each target $e_i^j$.
  Then set all target queries for shapes in $S$ to $\bot$ and introduce node shape $s_{N_0}$ with constraint $\phi_{N_0} =\ \geqslant_i ne_1^1.\phi_{s_N^1} \wedge \ldots \wedge \geqslant_i ne_n^m.\phi_{s_N^n} \wedge \geqslant_1 ee_1^1.\geqslant_1 (e_1^1 \wedge \phi_{s_E^1}) \wedge \ldots \wedge \geqslant_1 ee_n^m.\geqslant_1 (e_n^m \wedge \phi_{s_E^n})$.

  \hfill$\square$
\end{proof}
On the basis of this transformation, we can redefine strictly faithful assignments.
\begin{definition}[Strictly Faithful Assignment for Graphs with a Single Target Node]\label{def:faithful_assignment2}
  Let $s_{N_0}$ be the shape and $n_0$ the node constructed by \Cref{prop:singletarget} as the single target node.
  An assignment $\shapeassignment$ for a graph $G = (N,E,\rho,\lambda,\sigma)$ and a set of shapes $S$ is strictly faithful, if and only if:
  \begin{enumerate}
    \item $\forall\ s_N(n) \in \textrm{atoms}(G,S) : \Sigma(s_N(n)) = \evaln{\phi_N}$
    \item $\forall\ s_E(e) \in \textrm{atoms}(G,S) : \Sigma(s_E(e)) = \evale{\phi_E}$
    \item $\shapeassignment(s_{N_0}(n_0)) = 1$
  \end{enumerate}
\end{definition}
The validation problem $\validation$ for validation of property graphs with respect to a set of \pshacl shapes is defined as follows.
\begin{definition}[Validation]
    The problem of validating a property graph $G$ with respect to a set of shapes $S$ (such that in $S$ there is exactly one shape $s_{N_0}$ with a target query different from $\bot$ that targets node $n_0$, which can be constructed via \Cref{prop:singletarget} for any graph and set of shapes) is defined as $\valid{G}{S}{s_{N_0}(n_0)}$.
\end{definition}
We first show that $\validation$ is in NP.
\begin{theorem}
    \label{theorem:innp}
    \validation is in NP.
\end{theorem}
\paperonly{\begin{proof}[Sketch]
    In order to show that $\valid{G}{S}{s_{N_0}(n_0)}$ is in NP, we first construct, in polynomial time, an instance $\valid{G'}{S'}{s_{N_0}(n_0)}$ which is true if and only if $\valid{G}{S}{s_{N_0}(n_0)}$ is true, and $S'$ does not contain any path expressions (except for $l_E$) and each constraint in $S'$ has at most one operator.
    We assume an oracle for a strictly faithful assignment of such an instance $\valid{G'}{S'}{s_{N_0}(n_0)}$.
    Then we can, for each $s \in S'$, compute $\evaln{\phi_S}$ for each $n \in N$ and $\evale{\phi_s}$ for each $e \in E$ in polynomial time in $\vert \Sigma \vert + \vert G' \vert + \vert S' \vert$.
    \hfill$\square$
\end{proof}}
\reportorpaper{
The proof can be found in \Cref{sec:appa}.
}{
The complete proof can be found in an extended version of this work available on arXiv.
}
We next follow NP-hardness from the NP-hardness of $\mathcal{L}$.
\begin{corollary}
   \label{cor:nphard}
   RDF graph validation with $\mathcal{L}$, which is equivalent to SHACL, is clearly reducible to \pshacl validation over property graphs, since RDF graphs can be trivially represented in property graphs and constraints in $\mathcal{L}$ are a subset of \pshacl constraints.
   According to~\cite{DBLP:conf/semweb/CormanRS18}, $\mathcal{L}$ is NP-hard.
   Therefore, \pshacl is also NP-hard.
\end{corollary}
Then we can also conclude that $\validation$ for \pshacl is NP-complete.
\begin{corollary}
    \validation is NP-complete, since it is both NP-hard (shown in \Cref{cor:nphard}) and in NP (shown in \Cref{theorem:innp}).
\end{corollary}
We only consider the combined complexity here, even though graphs are typically significantly larger than sets of shapes.
However, from this we infer that validation for a fixed set of shapes (data complexity) and a fixed graph (constraint complexity) are also NP-complete, since they are already NP-complete for $\mathcal{L}$ as shown in~\cite{DBLP:conf/semweb/CormanRS18}, and combined complexity of validation for \pshacl is in NP.
\section{Implementation}
\label{sec:implementation}

Drawing inspiration from an experimental feature of the SHaclEX~\cite{shaclex} implementation of ShEx~\cite{shex} and SHACL~\cite{shacl}, we implement a prototypical validator for \pshacl by encoding the validation problem as an answer set program. 
Answer set programming (ASP) allows for declarative implementations of NP-hard search problems, such as \pshacl validation with faithful assignments.
In particular, we rely on ASP for efficiently finding candidate assignments (in the worst-case considering all possible assignments), while deciding whether an assignment is faithful is a straightforward mapping of our validation semantics to ASP.

The implementation consists of three components: An encoding of property graphs and \pshacl shapes, both of which are straight-forward and can be generated from non-ASP representations.
A set of rules directly representing the validation semantics of \pshacl (\Cref{subsec:eval}).
And finally the search problem of finding faithful assignments.
With these components, an ASP solver (our implementation relies on Clingo\footnote{\url{https://potassco.org/clingo/}}) produces one (or more) faithful assignments for the graph and set of shapes (if any exist).

In addition to the ASP encoding, we also provide a surrounding set of tools, including a concrete syntax for \pshacl shapes and a corresponding parser, as well as a tool for extracting and encoding Neo4j\footnote{\url{https://neo4j.com/}} instances.
The graph encoding is based on the Neo4j JSON export format and therefore straight-forward to replicate for other property-graph stores.
The tool suite is available on GitHub\footnote{\therepo}, including further documentation and examples.
\reportorpaper{
More details about the ASP encoding and a demonstration can be found in \Cref{sec:appb}.
}{
More details about the ASP encoding and a short demonstration can be found in the extended version of this work.
}

\subsection{Towards Practical Implementations of \pshacl}

Our implementation is well-suited as a reference implementation, for experimenting with \pshacl examples, and for validating smaller-sized graphs.
For large-scale graphs, the explicit ASP encoding of the data graph may be too inefficient, both in terms of runtime and memory requirements. 
Instead, efficient validation demands an implementation operating directly on a specific property-graph store.
Such an implementation could, for example, aim to replicate the resolution approach of an ASP solver for finding candidate assignments and evaluate the validation procedure directly on the graph.
For simplified SHACL shapes that do not include recursive shape references, efficient validation approaches are well-known and widely used in real-world SHACL implementations.
These approaches, operating on graph stores directly, could be applied for \pshacl as well.
Another alternative would be to adapt validation over SPARQL endpoints~\cite{DBLP:conf/semweb/CormanFRS19a} for Cypher and \pshacl instead.
Indeed, neosemantics~\cite{neosemantics} relies on Cypher for the validation of SHACL over RDF graphs encoded as property graphs.
Such an approach, as is also shown by \cite{DBLP:conf/semweb/CormanFRS19a}, can be extended to validate recursive shapes by inclusion of a SAT solver.
\section{Related Work}
\label{sec:relatedwork}

There are a number of schema languages for property graphs in proprietary implementations of graph databases.
For instance, the data definition language for Cypher~\cite{DBLP:conf/sigmod/FrancisGGLLMPRS18} described in the Neo4j manual~\cite{neoschema} allows for simple constraints regarding the existence or uniqueness of properties.
For TigerGraph~\cite{DBLP:journals/corr/abs-1901-08248}, a similar implementation exists.
However, these systems lack a formal description, making their expressiveness, features and complexity hard to assess.

Only a small number of property-graph schema languages have been formally defined.
In~\cite{DBLP:conf/grades/HartigH19}, the GraphQL~\cite{graphql} schema language is used to define restrictive property-graph schemas, where for each node label a GraphQL object type can be defined.
This allows constraining the existence of certain properties, edges, and properties on these edges via field definitions of the object types.
The schemas are closely tied to node labels, meaning the approach does not allow for the validation of edges as individual entities, which is crucial for validating metadata annotations across an entire graph.
The approach also omits other elements supported by \pshacl, such as negation, qualified number restrictions and path expressions in number restrictions or equality constraints.
Validation with constraints that are associated to labels can be emulated with \pshacl target queries.
Graph validation with this approach is shown to be in $\textrm{AC}_0$.

\cite{DBLP:conf/amw/Angles18} defines property graph schemas, also focusing on node and edge types on the basis of labels.
In particular, schemas allow for restricting the data types of specific properties on nodes and edges, as well as the edges allowed between node types. 
More advanced constraints are mentioned, but not formally defined.
In general, this approach only provides a small subset of the features of \pshacl.

While shape-based validation approaches such as SHACL~\cite{shacl} and ShEx~\cite{shex} exit for validating RDF graphs, to the best of our knowledge no shape-based validation language for property graphs has been formally defined until now.
A syntactic construct for SHACL validation of RDF* (and other reification-based RDF extensions) has been proposed in an unofficial draft proposal~\cite{shacldraft}, though no semantics has been specified.
The \texttt{reifiableBy} construct allows constraining an edge via a node shape for provenance annotations.
The approach is similar to our notion of edge shapes and our semantics can be applied, as long as graphs are restricted to property graphs (\ie edge properties are restricted to a given set of value domains).
Finally, there exists an extension for Neo4j which implements SHACL validation for RDF subsets of property graphs~\cite{neosemantics}.
\section{Concluding Remarks}
\label{sec:summary}

We present \pshacl, a shape language extending SHACL for validating property graphs.
We define the semantics of this language based on the notion of faithfulness of partial assignments and are therefore able to support shape references and negation.
Despite the addition of property-graph specific constructs, such as edge shapes that target edges with identities, the complexity of validating graphs against sets of \pshacl shapes does not increase when compared to SHACL.
The validation problem remains NP-complete.

As future work, we plan to investigate the satisfiability problem of \pshacl shapes and then further utilize these results to define a validation approach for property-graph queries.
We are also interested in extending \pshacl with the unique features introduced by G-CORE, in particular first-class paths, and RDF*, in particular triples in object position of other triples.
 



\bibliographystyle{splncs04}
\bibliography{paper}


\appendix
\reportonly{\newpage}
\reportonly{\section{Proofs}
\label{sec:appa}

\subsection{\validation is in NP}

\begin{theorem}
    \label{theo:polypath}
    For a path expression $p$ and a property graph $G$, deciding whether the node $n'$ is in $\evalpath{p}$ is possible in polynomial time in $\vert G \vert$. This follows from~\cite{DBLP:conf/semweb/KostylevR0V15}.
\end{theorem}

\begin{lemma}
    \label{lem:nopath}
    The validation problem $\valid{G}{S}{s_{N_0}(n_0)}$ can be transformed in polynomial time to $\valid{G'}{S'}{s_{N_0}(n_0)}$ such that $\valid{G}{S}{s_{N_0}(n_0)}$ is true if and only if $\valid{G'}{S'}{s_{N_0}(n_0)}$ is true, and $S'$ contains only path expressions of the form $l_E$.
\end{lemma}

\begin{proof}

    Let $p$ be a path expression, $G=(N,E,\rho,\lambda,\sigma)$ a property graph and $\shapeassignment$ a faithful assignment.
    $\evalpath{p}$ is the set of all nodes, such that there is a path from $n$ to $n' \in \evalpath{p}$ in $G$.
    For any given pair $(n,n') \in N \times N$ it can be decided in polynomial time in $\vert G \vert$, whether $n' \in \evalpath{p}$ (\Cref{theo:polypath}).

    Let $\valid{G}{S}{s_{N_0}(n_0)}$ be a validation problem and let $P$ be the set of all path expressions that appear in constraints in $S$.
    For each $p \in P$, let $l_p$ be fresh, unique edge label and $e_{p,n,n'}$ a fresh, unique edge from $n$ to each $n' \in \evalpath{p}$.
    
    Let $G'=(N,E',\rho',\lambda',\sigma)$ be a graph defined with $E'= E \cup \{e_{p,n,n'} \mid p \in P, n' \in \evalpath{p}\}$, $\rho' = \rho \cup \{e_{p,n,n'} \mapsto (n,n') \mid p \in P \land n' \in \evalpath{p}\}$ and $\lambda' = \lambda \cup \{ e_{p,n,n'} \mapsto l_p \mid p \in P, n' \in \evalpath{p}\}$.
    Given the previous statement, $G'$ can be computed in polynomial time.

    Let $S'$ be a new constraint set constructed from $S$, where each occurrence of path $p \in P$ is replaced by $l_p$ and $\shapeassignment'$ the same assignment as $\shapeassignment$, albeit using the equivalent shapes in $S'$.
    Then it is true, that for each $(n,n') \in N \times N$ the following equivalence holds: $\evalpath{p} = \evalpath[\shapeassignment'][n][G']{e_p}$.
    From this we can follow immediately that for each $n \in N$ and constraint $\phi$ in $S$, $\evaln{\phi} = \evaln[\shapeassignment'][n][G']{\phi}$, which means that $\valid{G}{S}{s_{N_0}(n_0)}$ is true, if and only if $\valid{G'}{S'}{s_{N_0}(n_0)}$ is true.

    \hfill$\square$
\end{proof}

\begin{lemma}
    \label{lem:oneop}
    The validation problem $\valid{G}{S}{s_{N_0}(n_0)}$ (where no constraints contain any path expressions except of the form $l_E$) can be transformed, in polynomial time in $\vert S \vert$ to a problem $\valid{G}{S'}{s_{N_0}(n_0)}$ such that $\valid{G}{S}{s_{N_0}(n_0)}$ is valid if and only if\\$\valid{G}{S'}{s_{N_0}(n_0)}$ is valid, and no constraint in $S'$ contains any path expressions other than of the form $l_E$ or more than one operator.
\end{lemma}

\begin{proof}

    We transform $S$ to $S'$ by introducing fresh shape names for each subformula in $S$.
    We call the function that transforms a shape $s \in S $ $\textrm{normalize}(s)$, defined as follows:
    If a constraint $\phi$ (either node or edge constraint) in $S$ has only one operator ($\neg$,$\wedge$,$\geqslant_i p.\phi_N$,$\geqleft{i}{\phi_E}$,$\geqright{i}{\phi_E}$ or $\Rightarrow \phi_N$,$\Leftarrow \phi_N$) where again $\phi$ can be either $\phi_N$ or $\phi_E$, then it remains the same in $S'$.
    Otherwise, we define the function $\textrm{fold}(s_i \phi_{s_i})$ that transforms shapes $s_i$ and their constraints $\phi_{s_i}$ recursively:
    \begin{itemize}
        \item $\phi_{s_i} = \phi_1 \wedge \phi_2$ then $\{s_i \mapsto s_1 \wedge s_2, s_1 \mapsto \phi_1, s_2 \mapsto \phi_2\}$
        \item $\phi_{s_i} = \neg \phi$ then $\{s_i \mapsto \neg s_1, s_1 \mapsto \phi\}$
        \item $\phi_{s_i} = \geqleft{i}{\phi_E}$ then $\{s_i \mapsto \geqleft{i}{s_1}, s_1 \mapsto \phi_E\}$
        \item $\phi_{s_i} = \geqright{i}{\phi_E}$ then $\{s_i \mapsto \geqright{i}{s_1}, s_1 \mapsto \phi_E\}$
        \item $\phi_{s_i} = \geqslant_i p.\phi_N$ then $\{s_i \mapsto \geqslant_i p.s_1, s_1 \mapsto \phi_N\}$
        \item $\phi_{s_i} = \Rightarrow \phi_N$ then $\{s_i \mapsto \Rightarrow s_1, s_1 \mapsto \phi_N\}$
        \item $\phi_{s_i} = \Leftarrow \phi_N$ then $\{s_i \mapsto \Leftarrow s_1, s_1 \mapsto \phi_N\}$
    \end{itemize}

    From the definition of constraint evaluation, it follows immediately that for each graph $G = (N,E,\rho,\lambda,\sigma)$, faithful assignment $\shapeassignment$, $n \in N$ and $e \in E$ and for each shape $s \in S$ with constraint $\phi_s$ and $s' = \textrm{normalize}(s)$ with constraint $\phi_{s'}$ that $\evaln{s} = \evaln[\shapeassignment']{s'}$ or $\evale{s} = \evale[\shapeassignment']{s'}$, when $\shapeassignment'$ is $\shapeassignment$ modified so that each node or edge assigned a shape $s \in S \cap S'$ is also assigned the respective shapes introduced via $\textrm{fold}$.
    Therefore, when transforming $S$ to $S'$ trough $\textrm{fold}$, we obtain a problem \\$\valid{G}{S'}{s_{N_0}(n_0)}$ that is true if and only if $\valid{G}{S}{s_{N_0}(n_0)}$.

    We next show, that this transformation is possible in polynomial time in $\vert S \vert$.
    When recursively applying fold to set of shapes $S_0,S_1,\ldots$ such that $S_0 = S$ and $S_{k+1} = \textrm{norm}(S_k)$, given the above definition of $\textrm{norm}$, in each step the number of operators per shape is either one, or decreases.
    Therefore, $\textrm{norm}$ must reach a fixed point $S_n$ for each shape in $S_0$.
    Since each application introduces at most 2 new shape names, $n$ in $S_n$ is bounded by $2 \times \sum_{s_i \in S} \textrm{operators}(s_i) = O(\vert S \vert)$, where $\textrm{operators}{s_i}$ is the number of operators in the constraint of shape $s_i$.
    For the same reason, the size of $S_n$ is bounded by $\vert S \vert + 2 \times \sum_{s_i \in S} \textrm{operators}(s_i)$.
    Finally, for each $s_i$ with $0 \leq i \leq n$, following from the definition of $\textrm{normalize}$ above, $\textrm{normalize}(s_i)$ can be computed in linear time $O(\vert S \vert)$.
    Therefore, $S_n$ can be computed in $O(\vert S \vert^2)$.

    \hfill$\square$
\end{proof}

\begin{lemma}
    \label{lem:validation}
    For the validation problem $\valid{G}{S}{s_{N_0}(n_0)}$ with no path expressions other than $l_E$ or constraint with more than one operator in any constraint in $S$ it can be decided in polynomial time in $\vert \Sigma \vert + \vert G \vert + \vert S \vert$ whether the assignment $\Sigma$ is strictly faithful.
\end{lemma}

\begin{proof}

    Let $\valid{G}{S}{s_{N_0}(n_0)}$ be a validation problem, where each constraint in $S$ has one or less operators and no path expressions other than $l_E$.
    Let $\Sigma$ be a faithful assignment.
    In order to show that $\Sigma$ is indeed faithful, it suffices to verify that $\Sigma(s(n)) = \evaln{\phi_s}$ for each atom $s(n) \in \textrm{atoms}_N(G,S)$ and $\Sigma(s(e)) = \evale{\phi_s}$ for each atom $s(e) \in \textrm{atoms}_E(G,S)$ (where $\phi_s$ is the constraint for shape $s$).
    We now show that this can indeed be verified in polynomial time in $\vert \Sigma \vert + \vert G \vert + \vert \phi_s \vert$, for which, since $\vert \phi_s \vert < \vert S \vert$ follows that $\vert \Sigma \vert + \vert G \vert + \vert \phi_s \vert$.

    For verification of $\Sigma(s(n)) = \evaln{\phi_s}$ we have the following cases. We omit the trivial cases where $\top$, $\neg \top$, $n$ or $\neg n$ could be used (e.g., in $\phi_s = \geqslant_i e.\top$).
    \begin{itemize}
        \item $\phi_s = \top$, then verify $\Sigma(s(n)) = 1$ in $O(|\phi_s\vert + \vert \Sigma \vert)$
        \item $\phi_s = \neg\top$, then verify $\Sigma(s(n)) = 0$ in $O(|\phi_s\vert + \vert \Sigma \vert)$
        \item $\phi_s = s'$, then verify $\Sigma(s(n)) = \Sigma(s'(n))$ in $O(|\phi_s\vert + \vert \Sigma \vert)$
        \item $\phi_s = \neg s'$, then verify $\Sigma(s(n)) = \Sigma(s'(n))$ in $O(|\phi_s\vert + \vert \Sigma \vert)$
        \item $\phi_s = s_1 \wedge s_2$, then verify $\Sigma(s(n)) = \textrm{min}\{\Sigma(s_1(n)), \Sigma(s_2(n))\}$ in $O(|\phi_s\vert + \vert \Sigma \vert)$

        \item $\phi_s = n'$, then verify $\Sigma(s(n)) = 1$ in $O(|\phi_s\vert + \vert \Sigma \vert)$
        \item $\phi_s = \neg n'$, then verify $\Sigma(s(n)) = 1$ in $O(|\phi_s\vert + \vert \Sigma \vert)$

        \item $\phi_s = \odot (l_{e_1},l_{e_2})$, then check if there are $e_1,e_2 \in E$ (where $e_1 \neg e_2$) with $\lambda(e_1) = e_{e_1}$ and $\lambda(e_2) = e_{e_2}$ and $\rho(e_1) = (n, n_1)$ and $\rho(e_2) = (n, n_2)$ and verify $\Sigma(s(n)) = 1$. Otherwise verify $\Sigma(s(n)) = 0$.
        The whole check is in $O(\vert \phi_s \vert + \vert G^2 \vert + \vert \Sigma \vert)$.

        \item $\phi_s = \odot (k_1,k_2)$, then verify $\Sigma(s(n)) = 1$ in $O(\vert \phi_s \vert + \vert \Sigma \vert)$
        \item The case $\phi_s = \odot (e_1, k_1, e_2, k_2)$ is a combination of the two previous cases

        \item $\phi_s = \geqslant_i l_e.s' $, then let $N_e = \{n' \mid e \in E, l_e \in \lambda(e), \rho(e) = (n, n') \}$, which can be computed in $O(\vert G \vert)$.
        Verify that $\vert N_e \vert < n$, in $O(\vert G \vert + \vert \phi_s \vert)$.
        If so, verify $\Sigma(s(n)) = 0$ in $O(\vert \phi_s \vert + \vert \Sigma \vert)$.
        Otherwise, check if there are more than $i$ $n' \in N_e$ such that $\Sigma(s(n')) = 1$ in $O(\vert \phi_s \vert + \vert \Sigma \vert)$.
        If there are fewer, verify $\Sigma(s(v)) = 0$ in in $O(\vert \phi_s \vert + \vert \Sigma \vert)$.

        \item The cases $\phi_s = \geqleft{i}{s}$ and $\phi_s = \geqright{i}{s}$ are similar to the previous one, excluding the second step involving $s'$.

        \item $\phi_s = \geqslant_i k.f$, then verify $\Sigma(s(n)) = 1$ in $O(|\phi_s\vert + \vert \Sigma \vert)$

    \end{itemize}

    \noindent
    For verification of $\Sigma(s(e)) = \evale{\phi_s}$ we have the following cases (again omitting trivial variants with $\top$, $n$ or negation thereof):
    \begin{itemize}
        \item $\phi_s = \Rightarrow s$, then verify $\Sigma(s(n')) = 1$ in $O(|\phi_s\vert + \vert \Sigma \vert)$ for $\rho(e) = (n,n')$.
        \item $\phi_s = \Leftarrow s$, then verify $\Sigma(s(n')) = 1$ in $O(|\phi_s\vert + \vert \Sigma \vert)$ for $\rho(e) = (n,n')$.
        \item We omit remaining cases, which are equivalent to node constraints.
    \end{itemize}

    \hfill$\square$
\end{proof}

\noindent
Now we can prove \textbf{\Cref{theorem:innp}}.

\begin{proof}
    Let $G_0$ be a property graph and $S_0$ a set of shapes such that $s_{N_O}$ is the only shape that has a target $n_0$ and the only target query different from $\bot$.

    \begin{enumerate}
    \item Applying \Cref{lem:nopath} we can, in polynomial time, construct a graph $G_1$ and set of shapes $S_1$ such that $\valid{G_1}{S_1}{s_{N_0}(n_0)}$ iff $\valid{G_0}{S_0}{s_{N_0}(n_0)}$ and the constraints in $S_1$ do not contain any property paths.

    \item Applying \Cref{lem:oneop} we can then, in polynomial time, construct a graph $G_2$ and set of shapes $S_2$ such that $\valid{G_2}{S_2}{s_{N_0}(n_0)}$ iff $\valid{G_1}{S_1}{s_{N_0}(n_0)}$ and each constraint in $S_2$ contains one or less operators and no property path expressions.

    \item Therefore, in order to show that $\valid{G_0}{S_0}{s_{N_0}(n_0)}$ holds, we can show that $\valid{G_2}{S_2}{s_{N_0}(n_0)}$.
    Assumption: $\valid{G_2}{S_2}{s_{N_0}(n_0)}$ is valid. 
    Then there exists a strictly faithful assignment $\Sigma$.
    Given that $\vert\textrm{atoms}(G,S)\vert=\vert N \vert\times\vert E \vert\times\vert S \vert \leqslant \vert G \vert \times \vert S \vert$, $\Sigma$ can be encoded as a string with size polynomial in the encoding of $\valid{G_2}{S_2}{s_{N_0}(n_0)}$.

    \item Given an oracle which, for a given $\valid{G_2}{S_2}{s_{N_0}(n_0)}$, returns a $\Sigma$, it can be verified via \Cref{lem:validation} in time polynomial in $\vert \Sigma \vert + \vert G_2 \vert + \vert S_2 \vert$ that $\Sigma$ is indeed a strictly faithful assignment.
    Therefore it directly follows that \validation is in NP.
    \end{enumerate}

    \hfill$\square$
\end{proof}}
\reportonly{\section{Demo: \pshacl Tool Suite}
\label{sec:appb}

This appendix gives a short overview of high-level aspects of the ASP-based implementation of \pshacl.
More concrete and technical usage instructions can be found as part of the documentation in the GitHub repository\footnote{\therepo}.

\subsection{High-Level Overview}

Executing validation with a graph instance and with a set of shapes either outputs the result \texttt{UNSATISFIABLE} or \texttt{SATISFIABLE} plus one faithful assignment (\ie a mapping for all nodes/edges and shapes, as well as the values \texttt{yes}, \texttt{no} or \texttt{maybe}) such that each target node is assigned \texttt{yes} and the assigned value corresponds to the result of evaluating the validation function for the respective node or edge. 
It is also possible to find all faithful assignments, though this may take a long time to compute for larger graphs.

To this end, the implementation (see \Cref{fig:coreasp} for an excerpt) relies essentially on the \texttt{assignE/3} as well as \texttt{assignN/3} predicates for finding faithful assignments such that each node/edge and shape is either assigned \texttt{yes}, \texttt{no} or \texttt{maybe} (lines 1-4) and predicates \texttt{satisfiesE/3} and \texttt{satisfiesN/3}, which encode the validation semantics of \pshacl (lines 6-7).
We also ensure that target nodes are assigned \texttt{yes} (lines 9-10).

\begin{figure}
    \begin{lstlisting}[style=asp,numbers=left]
assignN(N,S,yes) | assignN(N,S,no) | assignN(N,S,maybe) 
    :- node(N), nodeshape(S).
assignE(E,S,yes) | assignE(E,S,no) | assignE(E,S,maybe) 
    :- edge(E), edgeshape(S).

assignN(N,S,R) :- nodeshape(S,C,_), satisfiesN(N,C,R).
assignE(E,S,R) :- edgeshape(S,C,_), satisfiesE(E,C,R).

:- targetN(N,S), not assignN(N,S,yes).
:- targetE(E,S), not assignE(E,S,yes).
    \end{lstlisting}
    \caption{Outlining the core approach for finding faithful assignments.}
    \label{fig:coreasp}
\end{figure}

\subsection{Encoding of Property Graphs}

There are two ways of defining property graphs for validation.
Internally, property graphs are encoded in ASP.
This coding can be used directly to define property graphs for validation.
\Cref{fig:propertygraphencoding} shows an example, involving the three relevant predicates \texttt{edge/3}, \texttt{label/2} and \texttt{property/3} corresponding to $\rho$, $\lambda$ and $\sigma$ in the formalization.
Note, that $N$ and $E$ (nodes and edges) need not be explicitly defined.
The encoding has some minor restrictions: In particular, properties and labels must start with lower case letters.
This is due to the fact that in ASP uppercase letters indicate variables.

\begin{figure}
    \begin{lstlisting}[style=asp,numbers=left]
// Define all edges in the graph.
edge(100, 200, 101).
edge(100, 201, 102).
edge(102, 202, 100).
edge(102, 203, 101).

// Define labels for nodes and edges.
label(100, employee).
label(100, person).
label(101, company).
label(200, worksFor).
label(201, colleagueOf).
// ...

// Define properties for nodes and edges.
property(100, name, string("Tim Canterbury")).
property(100, age, integer(30)).
// ...
    \end{lstlisting}
    \caption{Excerpt of the encoding of the property graph $\examplegraph$ from \Cref{fig:intro} in ASP.}
    \label{fig:propertygraphencoding}
\end{figure}

In addition to property graphs directly specified using the formal ASP encoding, real-world property graphs can be validated.
For this purpose, the tool suite provides functionality for exporting and converting Neo4j instances to the ASP encoding.
The conversion relies on the Cypher JSON export format, but can be readily replicated for other graph stores.
Due to the previously mentioned restriction, our conversion tool converts all identifiers (labels and properties) to lower case.
More information on how to export Neo4j instances can be found in the documentation.

\subsection{Encoding of Shapes}

The ASP encoding of shapes is, for the most part, straightforward.
In addition to the \texttt{nodeshape/3} (and \texttt{edgeshape/3}) predicates, each component constraint and path must be explicitly listed, so that rules in the semantics are safe and can be grounded.
When using the concrete syntax (see next section), this is taken care of automatically.
\Cref{fig:encodedshapes} shows an example ASP encoding of the shape $s_1$ introduced in \Cref{ex:colleague}.
The concrete syntax of this shape is discussed in the following section.

\begin{figure}
    \begin{lstlisting}[style=asp,numbers=left]
constraint(greaterEq(label(colleagueOf),label(person),1)).
constraint(label(person)).
path(label(colleagueOf)).

nodeshape(s1,                                            // shape name
          greaterEq(label(colleagueOf),label(person),1), // constraint
          label(employee)).                              // target
    \end{lstlisting}
    \caption{Example for the ASP encoding of shapes, using the shape $s_1$ from \Cref{ex:colleague}.}
    \label{fig:encodedshapes}
\end{figure}

\subsection{Concrete Syntax of Shapes}

The tool suite also provides a concrete syntax for shapes, which is the strongly recommended way of defining shapes, since parser errors, for example, are an indication of ill-defined shapes, that may go unnoticed when using the ASP encoding directly.
\Cref{fig:concreteshapes} shows the example shape from \Cref{fig:encodedshapes} using the concrete syntax.
More details (as well as the grammar of the concrete \pshacl syntax) are available on GitHub.

\begin{figure}
    \centering 
    \begin{lstlisting}[style=progs,numbers=left]
NODE s1 [:employee] {
    >= 1 :colleagueOf . :person
};
    \end{lstlisting}
    \caption{Example for the concrete syntax of shapes, using the shape $s_1$ from \Cref{ex:colleague}.}
    \label{fig:concreteshapes}
\end{figure}

\subsection{Example Encoding of Validation Semantics}

Finally, we demonstrate how \texttt{satisfiesN/3} encodes the validation semantics, using $\wedge$ for node constraints as an example.
\Cref{fig:aspand} shows the direct mapping of validation semantics, in this case the recursive validation of both components \texttt{C1} and \texttt{C2} (line 2).
We first ensure that \texttt{N} and \texttt{and(C1,C2)} can be grounded (line 1).
Finally, just as in the formal definition in \Cref{fig:nconstrainteval}, we find the minimum for \texttt{yes}, \texttt{no} or \texttt{maybe} in \texttt{R} (line 3).

\begin{figure}
    \begin{lstlisting}[style=asp, numbers=left]
satisfiesN(N,and(C1,C2),R) :- node(N), constraint(and(C1,C2)),
                              satisfiesN(N,C1,R1), satisfiesN(N,C2,R2),
                              min(R1,R2,R).
    \end{lstlisting}
    \caption{Encoding of constraint $\wedge$ for node constraints.}
    \label{fig:aspand}
\end{figure}}

\end{document}